\documentclass[12pt]{article}

\usepackage[utf8]{inputenc}
\usepackage[T1]{fontenc}
\usepackage{eucal}
\usepackage{microtype}
\usepackage{csquotes}
\usepackage{amsthm}
\usepackage{mathtools}
\usepackage{amssymb,amsfonts,stmaryrd}
\usepackage{tikz}
\usetikzlibrary{cd}
\usepackage[citestyle=numeric-comp,bibstyle=ieee,sorting=nyt,sortcites]{biblatex}
\usepackage{enumitem}
\usepackage[unicode=true, pdfusetitle, colorlinks=true, allcolors=blue, bookmarks=false]{hyperref}
\usepackage{authblk} 

\theoremstyle{plain}
\newtheorem{theorem}{Theorem}
\newtheorem*{theorem*}{Theorem}
\newtheorem{lemma}[theorem]{Lemma}
\newtheorem*{lemma*}{Lemma}
\newtheorem{proposition}[theorem]{Proposition}
\newtheorem*{proposition*}{Proposition}
\newtheorem{corollary}[theorem]{Corollary}  
\newtheorem*{corollary*}{Corollary}
\theoremstyle{definition}
\newtheorem{definition}{Definition}
\newtheorem*{definition*}{Definition}
\newtheorem{example}{Example}
\newtheorem*{example*}{Example}
\theoremstyle{remark}
\newtheorem{remark}{Remark}
\newtheorem*{remark*}{Remark}

\newtheorem*{conjecture*}{Conjecture}

\newtheorem*{problem*}{Problem}

\newcommand*{\RR}{\mathbb{R}}

\newcommand*{\dd}{\mathrm{d}}
\DeclareMathOperator{\Id}{Id}
\newcommand*{\contr}[1]{\iota_{#1}}
\newcommand*{\liedv}[1]{\mathcal{L}_{#1}}

\newcommand*{\parder}[2]{\frac{\partial#1}{\partial #2}}
\newcommand*{\der}[2]{\frac{\dd #1}{\dd #2}}
\DeclareMathOperator{\Leg}{Leg}
\newcommand{\incl}{\iota}

\bibliography{biblio}

\AtEveryBibitem{
  \ifentrytype{online}{}{
    \clearfield{url}
  }
}

\DeclareSourcemap{
  \maps[datatype=bibtex]{
    \map[overwrite=true]{
      \step[fieldset=urldate, null]
      \step[fieldset=language, null]
      \step[fieldset=address, null]
     \step[fieldset=addendum, null]
     \step[fieldset=pagetotal, null]
    }
  }
}

\title{Symmetries, constants of the motion and reduction of mechanical systems with external forces}
\author[1,2]{Manuel de León}
\author[1]{Manuel Lainz}
\author[1]{Asier López-Gordón\thanks{Author to whom correspondence should be addressed: \href{mailto:asier.lopez.gordon@csic.es}{asier.lopez.gordon@csic.es}}
}
\affil[1]{Instituto de Ciencias Matemáticas (CSIC-UAM-UC3M-UCM) \protect\\
Calle Nicolás Cabrera, 13-15, Campus Cantoblanco, UAM, 28049 Madrid, Spain}
\affil[2]{Real Academia de Ciencias Exactas, Físicas y Naturales\protect\\
Calle Valverde, 22, 28004, Madrid, Spain }

\date{\today}

\mathtoolsset{showonlyrefs}
\begin{document}

\maketitle

\begin{abstract}
\noindent
This paper is devoted to the study of mechanical systems subjected to external forces in the framework of symplectic geometry. We obtain a Noether's theorem for Lagrangian systems with external forces, among other results regarding symmetries and conserved quantities. We particularize our results for the so-called Rayleigh dissipation, i.e., external forces that are derived from a dissipation function, and illustrate them with some examples. Moreover, we present a theory for the reduction of Lagrangian systems subjected to external forces which are invariant under the action of a Lie group.
\end{abstract}

\section{Introduction}
In this paper, we study the geometry and symmetries of Hamiltonian and Lagrangian systems with external forces, focusing on the so-called systems with Rayleigh dissipation. 
Mechanical systems with external forces are usual in Engineering \cite{cantrijn_82,cantrijn_84,cantrijn_02}, but also can arise in a more sophisticated manner, for instance, after a process of reduction of a nonholonomic system with symmetries \cite{cantrijn_02,cantrijin_et_al_99,cortes_99}. As it is well-known (see Refs.~\cite{godbillon_69,MdL_Rodrigues_89}), external forces can be regarded as semibasic 1-forms on the tangent or cotangent bundle. 
Our approach is based on the symplectic structure obtained from a regular Lagrangian in the Lagrangian formulation as well as the geometry of the tangent bundle. There are other ways to treat with symmetries, for instance a variational approach like in Ref.~\cite{bahar_87}.

The main result when we are in presence of symmetries is the celebrated Noether theorem. See Ref.~\cite{noether_71} for the original paper by E. Noether (see also Refs.~\cite{kosmann_11,neeman_99}). In our case, in spite of the existence of a non-conservative external force, we are able to extend Noether theorem and, furthermore, to obtain new conserved quantities. Our approach is just an appropriate modification of the well-known results for conservative mechanical systems (that means with no external forces) \cite{marmo_mukanda_86,djukic_vujanovic_75,sarlet_cantrijn_81,MdL_Rodrigues_89,carinena_94,carinena_lopez_martinez_89,carinena_martinez_89,ferrario_90,MdL_DMdD_94,de_leon_symmetries_1995,de_leon_symmetries_1996,lunev_90,marwat_07,prince_83,prince_85,roman-roy_20,sarlet_83,sarlet_cantrijn_81}. So, we first define point-base symmetries (that is, those provided by vector fields on the configuration manifold $Q$), and then symmetries on the tangent bundle. 

There are other approaches that can be found in the previous literature and have some relation with ours. For instance, Cantrijn \cite{cantrijn_82} considers Lagrangian systems that depend explicitly on time, and defines a 2-form on $\RR\times TQ$ that depends on the Poincaré-Cartan 2-form of the Lagrangian and the semibasic 1-form representing the external force.  Alternatively, van der Schaft \cite{schaft_83,schaft_81} considers a framework steming from system theory, in which an ``observation'' manifold appears together with the usual state space, and obtains a Noether's theorem for Hamiltonian system in this frame. Other approaches using variational tools can be found in Ref.~\cite{bahar_87}. However, in our approach no additional structure or objects are introduced besides the proper external force. 

The paper is organized as follows. In Sections \ref{section_Hamiltonian} and \ref{section_Lagrangian} we review Hamiltonian and Lagrangian systems with external forces, respectively. In Section \ref{section_morphisms} we cover the relation between fibre bundle morphisms and semibasic 1-forms. 
In Section \ref{section_conserved_Lagrangian}, we present some (as far as we know) original results concerning symmetries and constants of the motion for mechanical systems with external forces. 
In Section \ref{section_conserved_Hamiltonian} we study the symmetries and constants of the motion in the Hamiltonian framework. We relate these symmetries with the ones obtained for Lagrangian systems in the previous section.
In Section \ref{section_Rayleigh} we particularize the results of the previous section for the Rayleigh dissipation. Classically \cite{rayleigh_1871,goldstein_87,gantmakher_70}, only external forces that are linear on the velocities are regarded as examples of Rayleigh dissipation. However, following Lurie \cite{lurie_02} and Minguzzi \cite{minguzzi_15}, we consider a wider family of external forces as Rayleigh dissipation, namely forces that are derived from a dissipation function (which is not necessarily quadratic on the velocities). 
Finally, in Section \ref{section_reduction} we present a scheme for reduction in Lagrangian systems subjected to external forces which are invariant under the action of a Lie group.

\section{Hamiltonian systems subject to external\\ forces}\label{section_Hamiltonian}
An external force is geometrically interpreted as a semibasic 1-form on $T^*Q$. Let us recall \cite{MdL_Rodrigues_89,abraham_marsden_08,godbillon_69} that a 1-form $\gamma$ on $T^*Q$ is called \emph{semibasic} if 
\begin{equation}
    \gamma(Z)=0
\end{equation}
for all vertical vector fields $Z$.

\begin{remark}
    This definition can be extended to any fibre bundle $\pi: E\to M$. Indeed, a 1-form $\gamma$ on $E$ is called semibasic if
    \begin{equation}
        \gamma(Z)=0
    \end{equation}
    for all vertical vector fields $Z$ on $E$. If $(x^i,y^a)$ are fibred (bundle) coordinates, then the vertical vector fields are locally generated by $\{\partial/\partial y^a\}$. So $\gamma$ is a semibasic 1-form if it is locally written as
    \begin{equation}
        \gamma=\gamma_i(x,y) \dd x^i.
    \end{equation}
\end{remark}

A Hamiltonian system with external forces is given by a Hamiltonian function $H:T^*Q\to \mathbb{R}$ and a semibasic 1-form $\gamma$ on $T^*Q$. Let $\omega_Q=-\dd \alpha_Q$ be the canonical symplectic form of $T^*Q$.
Locally these objects can be written as

\begin{equation}
\begin{aligned}
&\alpha_Q=p_i\dd q^i,\\
&\omega_Q=\dd q^i\wedge \dd p_i,\\
&\gamma=\gamma_i(q,p) \dd q^i,\\
&H=H(q,p),
\end{aligned}
\end{equation}
where $(q^i,p_i)$ are bundle coordinates in $T^*Q$.

The dynamics of the system is given by the vector field $X_{H,\gamma}$, defined by
\begin{equation}
\contr{X_{H,\gamma}}\omega_Q=\dd H+\gamma.
\end{equation}
If $X_H$ is the Hamiltonian vector field for $H$, that is,
\begin{equation}
\contr{X_H}\omega_Q=\dd H, \label{Hamiltonian_VF}
\end{equation}
and $Z_\gamma$ is the vector field defined by
\begin{equation}
\contr{Z_\gamma}\omega_Q=\gamma,
\end{equation}
then we have
\begin{equation}
X_{H,\gamma}=X_H+Z_\gamma.
\end{equation}
Locally, the above equations can be written as 
\begin{equation}
\begin{aligned}
&X_H=\parder{H}{p_i}\parder{}{q^i}-\parder{H}{q^i}\parder{}{p_i} \label{Hamiltonian_VF_local},
\\
&\gamma=\gamma_i \dd q^i,\\
&Z_\gamma=-\gamma_i\parder{}{p_i},\\
&X_{H,\gamma}=\parder{H}{p_i}\parder{}{q^i}-\left(\parder{H}{q^i}+\gamma_i\right)\parder{}{p_i}.
\end{aligned}
\end{equation}
Then, a curve $(q^i(t), p_i (t)$ in $T^*Q$ is an integral curve of $X_{H,\gamma}$ if and only if it satisfies the forced motion equations
\begin{equation}
\begin{aligned}
&\der{q^i}{t}=\parder{H}{p_i},\\
&\der{p_i}{t}=-\left(\parder{H}{q^i}+\gamma_i\right).
\end{aligned}
\end{equation}
\section{Semibasic forms and fibred morphims}\label{section_morphisms}
Given a semibasic 1-form $\gamma$ on $TQ$, one can define the following morphism of fibre bundles \cite{MdL_Rodrigues_89,godbillon_69}:

\begin{equation}
\begin{aligned}
    &D_\gamma: TQ\to T^*Q,\\
    &\left\langle D_\gamma(v_q),w_q\right\rangle=\gamma(v_q)(u_{w_q}),
\end{aligned}
\end{equation}
for every $v_q,w_q\in T_qQ,\ u_{w_q}\in T_{w_q}(TQ)$, with $T\tau_Q(u_{w_q})=w_q$. In local coordinates, if
\begin{equation}
\gamma=\gamma_i(q,\dot q)\dd q^i,
\end{equation}
then
\begin{equation}
    D_\gamma(q^i,\dot q^i)=\left(q^i,\gamma_i(q^i,\dot{q}^i)\right).
\end{equation}
Here $(q^i, \dot{q}^i)$ are bundle coordinates in $TQ$.

Conversely, given a morphism of fibre bundles
\begin{center}
\begin{tikzcd}
D: TQ \arrow[rr] \arrow[rd, "\tau_q"'] &   & T^*Q \arrow[ld, "\pi_Q"] \\
                                                          & Q &                         
\end{tikzcd},
\end{center}
we define a semibasic 1-form $\gamma$ on $TQ$ by 
\begin{equation}
    \gamma_D(v_q)(u_{v_q})=\left\langle D(v_q),T\tau_Q(u_{v_q})\right\rangle,
\end{equation}
where $v_q\in T_qQ,\ u_{v_q}\in T_{v_q}(TQ)$.

If locally $D$ is given by
\begin{equation}
    D(q^i,\dot q^i)=(q^i,D_i(q,\dot q)),
\end{equation}
then
\begin{equation}
    \gamma_D=D_i(q,\dot q)\dd q^i.
\end{equation}
So there exists a one-to-one correspondence between semibasic 1-forms and fibred morphisms from $TQ$ to $T^*Q$.

\section{Lagrangian systems with external forces} \label{section_Lagrangian}
We shall now consider a Lagrangian system with Lagrangian function $L$ subjected to external forces. An external force is given by a semibasic 1-form $\beta$ on $TQ$. In bundle coordinates, we have
\begin{equation}
\beta=\beta_i(q,\dot{q})\dd q^i.
\end{equation}
If $L:TQ\to \RR$, then $\omega_L=-\dd \alpha_L$ is the \emph{Poincaré-Cartan 2-form}, where $\alpha_L=S^*(\dd L)$. Here, $S$ is the vertical endomorphism of $TQ$, which in local coordinates is given by
\begin{equation}\label{eq:canonical_endomorphism}
    S = \dd q^i \otimes \frac{\partial}{\partial \dot{q}^i},   
\end{equation}
hence,
\begin{equation}
    \omega_L = \dd q^i \wedge \dd \left(\frac{\partial L}{\partial \dot{q}^i}\right).
\end{equation}
Then, the dynamics is given by the vector field $\xi_{L,\beta}$ via the equation
\begin{equation}
\contr{\xi_{L,\beta}}\omega_L=\dd E_L+\beta, \label{dynamics_vector_Lagrangian}
\end{equation} 
where $E_L=\Delta(L)-L$ is the energy of the system and $\Delta$ is the \emph{Liouville vector field}: 
\begin{equation}
    \Delta = \dot{q}^i \frac{\partial}{\partial \dot{q}^i}.
\end{equation}
Here, we are assuming that $L$ is \emph{regular}, that is, the Hessian matrix
\begin{equation}
    (W_{ij})=\left(\frac{\partial^2 L}{\partial \dot{q}^i\partial \dot{q}^j}\right). \label{Hessian_Lagrangian}
\end{equation}
is invertible. It can be easily proven that $\omega_L$ is symplectic if and only if $L$ is regular \cite{MdL_Rodrigues_89}. Let $\xi_\beta$ be the vector field given by
\begin{equation}
\contr{\xi_{\beta}}\omega_L=\beta,
\end{equation}
and $\xi_L$ be the vector field given by
\begin{equation}
    \contr{\xi_L}\omega_L=\dd E_L, \label{Lagrangian_VF}
\end{equation}
then
\begin{equation}
\xi_{L,\beta}=\xi_L+\xi_\beta.
\end{equation}
We have
\begin{equation}
\xi_\beta=-\beta_iW^{ij}\parder{}{\dot{q}^j},
\end{equation}
where $(W^{ij})$ is the inverse matrix of $(W_{ij})$.
Then $\xi_{L,\beta}$ is a \emph{second order differential equation (SODE)}, meaning that,
\begin{equation}
S(\xi_{L,\beta})=S(\xi_L)=\Delta. \label{SODE}
\end{equation}
We know that
\begin{equation}
\xi_L=\dot{q}^i\parder{}{q^i}+\xi^i\parder{}{\dot{q}^i}, 
\end{equation}
where
\begin{equation}
\xi^i\parder{p_j}{\dot{q}^i}+\dot{q}^i\parder{p_j}{q^i}-\parder{L}{q^j}=0. \label{almost_EL}
\end{equation}
Then
\begin{equation}
\xi_{L,\beta}=\dot{q}^i\parder{}{q^i}+\left(\xi^i-\beta_j W^{ij}\right)\parder{}{\dot{q}^i}.
\end{equation}
Hence, a solution of $\xi_{L,\beta}$, $(q^i(t))$, satisfies
\begin{equation}
\begin{aligned}
&\der{q^i}{t}=\dot{q}^i,\\
&\der{\dot{q}^i}{t}=\xi^i-\beta_jW^{ji}.
\end{aligned}
\end{equation}
Therefore, from Eq.~\eqref{almost_EL}, we get
\begin{equation}
\ddot{q}^i\parder{p_j}{\dot{q}^i}+\dot{q}^i\parder{p_j}{q^i}-\parder{L}{q^j}+\beta_k W^{ki}\parder{p_j}{\dot{q}^i}=0.
\end{equation}
Since $p_j=\partial{L}/\partial \dot{q}^j$, the term $\partial p_j/\partial \dot{q}^i$ is equal to $W_{ji}$, and thus we finally obtain
\begin{equation}
\der{}{t}\left(\parder{L}{\dot{q}^i}\right)-\parder{L}{q^i}=-\beta_i. \label{Euler_Lagrange_eq}
\end{equation} 
If we construct the Legendre transform \cite{MdL_Rodrigues_89}
\begin{center}

\begin{tikzcd}
TQ \arrow[rr, "\Leg"] \arrow[rd, "\tau_q"'] &   & T^*Q \arrow[ld, "\pi_Q"] \\
                                                          & Q &                         
\end{tikzcd}
\end{center}
(and assume $L$ be \emph{hyperregular}, that is, $\Leg$ is a diffeomorphism), then we can define the external force $\gamma$ on $T^*Q$ by
\begin{equation}
\Leg^*\gamma=\beta.
\end{equation}
Thus $\xi_{L,\beta}$ and $X_{H,\gamma}$ are $\Leg$-related, that is,
$\Leg$ takes $\xi_{L,\beta}$ onto $X_{H,\gamma}$, where $H$ is defined by
\begin{equation}
H \circ \operatorname{Leg}=E_L. \label{Hamiltonian_Lagrangian}
\end{equation}

\begin{definition}
 In what follows, we will refer to the pair $(L, \beta)$ for a forced Lagrangian system given by a Lagrangian $L$ and a semibasic 1-form $\beta$. The corresponding vector field $\xi_{L,\beta}$, given by Eq.~\eqref{dynamics_vector_Lagrangian}, will be called \emph{forced Euler-Lagrange vector field}.
\end{definition}

\begin{remark}
    Take
    \begin{equation}
        \alpha_L=S^*(\dd L)=p_i \dd q^i,
    \end{equation}
    where $p_i = {\partial L}/{\partial \dot{q}^i}$,
    then $\alpha_L$ is a semibasic 1-form on $TQ$; and the corresponding fibred map is just the Legendre transform 
    \begin{equation}
        \Leg:TQ\to T^*Q.
    \end{equation}
\end{remark}

\section{Symmetries and constants of the motion in the Lagrangian description}\label{section_conserved_Lagrangian}

Let $f:TQ\to \RR$ be an arbitrary function and $\tau_Q:TQ\to Q$ the projection. Then the \emph{vertical lift} \cite{yano_ishihara_73,yano_almost_1967} of $f$ is a function $f^v:TQ\to Q$ given by
\begin{equation}
    f^v=f\circ \tau_Q.
\end{equation}
Any 1-form $\omega$ in $Q$ can be naturally regarded as a function on $TQ$, which we shall denote by $\iota \omega$.
If $X$ is a vector field on $Q$, its \emph{vertical lift} is the unique vector field $X^v$ on $TQ$ such that
\begin{equation}
    X^v(\iota \omega)= \left( \alpha(X)  \right)^v
\end{equation}
for every 1-form $\alpha$ on $Q$. The \emph{complete lift} of a function $f$ on $Q$ is the function $f^c$ on $TQ$ given by
\begin{equation}
    f^c=\iota (\dd f).
\end{equation}
The \emph{complete lift} of a vector field $X$ on $Q$ is the vector field $X^c$ on $TQ$ such that
\begin{equation}
    X^c (f^c) = \left( X(f)  \right)^c
\end{equation}
for every function $f$ on $Q$. If $X$ generates locally a 1-parameter group of transformations on $Q$, then $X^c$ generates the induced transformations on $TQ$ \cite{MdL_Rodrigues_89}.
Locally, if $X$ is given by
\begin{equation}
  X=X^i \frac{\partial  } {\partial q^i},
\end{equation}
then its \emph{vertical lift} is
\begin{equation}
   X^v= X^i \frac{\partial  } {\partial \dot{q}^i},
 \end{equation} 
and its \emph{complete lift} is 
\begin{equation}
  X^c= X^i \frac{\partial  } {\partial q^i}+\dot{q}^j \frac{\partial X^i} {\partial q^j} \frac{\partial  } {\partial \dot{q}^i}.
\end{equation}
Let $(L,\beta)$ be a Lagrangian system with Lagrangian function $L$ and external force $\beta$; denote by $\xi_{L,\beta}$ the corresponding forced Euler-Lagrange vector field. 
\begin{definition}
A function $f$ on $TQ$ is called a \emph{constant of the motion} (or a \emph{conserved quantity}) if $\xi_{L,\beta}(f)=0$.
\end{definition}

Suppose that, for a certain coordinate $q^i$, $\partial L/\partial q^i=\beta_i$.
 Then 
\begin{equation}
\der{}{t}\left(\parder{L}{\dot{q}^i}\right)=0,
\end{equation} 
and
$p_i=\partial L/\partial \dot{q}^i$ is a constant of the motion. This motivates the following theorem.

\begin{theorem}[Noether's theorem for dissipative systems]\label{Noether_th}
Let $X$ be a vector field on $Q$. Then $X^c(L)=\beta(X^c)$ if and only if $X^v(L)$ is a constant of the motion.
\end{theorem}

\begin{proof}
By Eq.~\eqref{dynamics_vector_Lagrangian}, we can write
\begin{equation}
  \begin{aligned}
  (\dd E_L+\beta)(X^c)=& \left(\contr{\xi_{L,\beta}} \omega_L  \right) (X^c)=-\dd \alpha_L(\xi_{L,\beta},X^c)\\
  =&-\xi_{L,\beta}(\alpha_L(X^c))+X^c(\alpha_L(\xi_{L,\beta}))+\alpha_L([\xi_{L,\beta},X^c]).
  \end{aligned}
\end{equation}
Now, since $\xi_{L,\beta}$ is a SODE, we have
\begin{equation}
\alpha_L(\xi_{L,\beta})=\contr{\xi_{L,\beta}}(S^*\dd L)=(S\xi_{L,\beta})L=\Delta L.  
\end{equation}
It is easy to see that $S X^c= X^v$. Moreover, $[\xi_{L,\beta},X^c]$ is a vertical vector field, and thus $S[\xi_{L,\beta},X^c]=0$. Then
\begin{equation}
  (\dd E_L+\beta)(X^c)=-\xi_{L,\beta}(X^v L)+X^c(\Delta L).
\end{equation}
On the other hand, we can write
\begin{equation}
  \dd E_L (X^c)=X^c(E_L)=X^c(\Delta L)-X^c(L).
\end{equation}
Combining these last two equations one deduces
\begin{equation}
  \xi_{L,\beta} (X^v L)=X^c (L) -\beta (X^c).
\end{equation}
In particular, the right-hand side vanishes if and only if the left-hand side does.
\end{proof}

\begin{definition} Consider the forced Lagrangian system $(L,\beta)$. Then
\begin{enumerate}[label=\roman*)]
\item A \emph{symmetry of the forced Lagrangian system} is a vector field $X$ on $Q$ such that $X^c(L)=\beta(X^c)$. 
\item A \emph{Lie symmetry} is a vector field $X$ on $Q$ such that $[X^c,\xi_{L,\beta}]=0$. 
\item A \emph{Noether symmetry} is a vector field $X$ on $Q$ such that $X^c(E_L)+\beta(X^c)=0$ and $\liedv{X^c}\alpha_L$ is exact.
\end{enumerate}
\end{definition}

\begin{proposition}\label{prop_Lie_sym}
 If $X$ is a vector field on $Q$ such that  
\begin{equation}
  \dd (\liedv{X^c}\alpha_L)=0,
\end{equation}
then $X$ is a Lie symmetry if and only if 
\begin{equation}
  \liedv{X^c}\beta= -\dd (X^c (E_L)).
\end{equation}
\end{proposition}
\begin{proof}
  Indeed,
  \begin{equation}
    \begin{aligned}
     \contr{[X^c,\xi_{L,\beta}]}\omega_{L}=& \liedv{X^c} (\contr{\xi_{L,\beta}}\omega_{L})- \contr{\xi_{L,\beta}} ( \liedv{X^c}\omega_{L}) \\
     =& \liedv{X^c} (\dd E_L+\beta)+ \contr{\xi_{L,\beta}} \dd ( \liedv{X^c}\alpha_{L}) \\
     =&\dd (X^c(E_L))+\liedv{X^c}\beta.
    \end{aligned}
  \end{equation}
  Since $\omega_L$ is non-degenerate, then $[X^c,\xi_{L,\beta}]$ vanishes if and only if $\contr{[X^c,\xi_{L,\beta}]}\omega_{L}$ does.
\end{proof} 

 \begin{proposition}
 A Noether symmetry is a Lie symmetry if and only if 
 \begin{equation}
   \contr{X^c}\dd \beta=0.
 \end{equation}
 \end{proposition}
 \begin{proof}
    Since $\liedv{X^c}\alpha_L$ is exact, it can be written as $\liedv{X^c}\alpha_L=\dd f$ for some function $f:TQ\to \RR$.
    Obviously, $\dd(\liedv{X^c}{\alpha_L})=\dd (\dd f)=0$. In addition,
    \begin{equation}
    \begin{aligned}
      \liedv{X^c} \beta&=\contr{X^c}(\dd \beta)+\dd (\contr{X^c}\beta)
      =\contr{X^c}(\dd \beta)+\dd(\beta(X^c))\\
      &=\contr{X^c}(\dd \beta)-\dd(X^c(E_L)).
      \end{aligned}
    \end{equation}
    By Proposition~\ref{prop_Lie_sym}, the result holds.
 \end{proof}

\begin{proposition}
Let $X$ be a vector field on $TQ$ such that
\begin{equation}
  \liedv{X^c}\alpha_L=\dd f,
\end{equation}
then $X$ is a Noether symmetry if and only if $f-X^v(L)$ is a conserved quantity. 
\end{proposition}
\begin{proof}
Indeed, 
\begin{equation}
\begin{aligned}
  \dd f&= \liedv{X^c}\alpha_L=\contr{X^c}(\dd \alpha_L)+\dd (\contr{X^c}\alpha_L)=\contr{X^c}(\dd \alpha_L)+\dd (\contr{X^c}S^* \dd L)
  \\
  &=\contr{X^c}(\dd \alpha_L)+\dd (\contr{SX^c}\dd L)=\contr{X^c}(\dd \alpha_L)+\dd (X^v L),
\end{aligned}
\end{equation}
so
\begin{equation}
  \contr{\xi_{L,\beta}}\contr{X^c}(\dd \alpha_L)=\contr{\xi_{L,\beta}}(\dd (f-X^vL))=\xi_{L,\beta}(f-X^v(L)),
\end{equation}
but
\begin{equation}
  \contr{\xi_{L,\beta}}\contr{X^c}\dd \alpha_L=
  \contr{X^c}\contr{\xi_{L,\beta}}\omega_L=\contr{X^c}(\dd E_L+\beta)=X^c(E_L)+\beta(X^c),
\end{equation}
and the result holds.
\end{proof}

Observe that this last proposition is a generalisation of Theorem \ref{Noether_th}. In other words, every symmetry of the forced Lagrangian system is a Noether symmetry. In fact, if $f$ is a constant function, clearly $X^v L$ is a conserved quantity. Moreover,
\begin{equation}
  \liedv{X^c}\alpha_L=0,
\end{equation}
so
\begin{equation}
\begin{aligned}
  0&=(\liedv{X^c}\alpha_L)(\xi_{L,\beta})=X^c\left(\alpha_L(\xi_{L,\beta})\right)-\alpha_L ([X^c,\xi_{L,\beta}])\\
  &=X^c(\Delta L)-S[X^c,\xi_{L,\beta}]L=X^c(\Delta L),
  \end{aligned}
\end{equation}
and thus,
\begin{equation}
  0=X^c(E_L)+\beta(X^c)=X^c(\Delta L) -X^c(L)+\beta(X^c)=-X^c(L)+\beta(X^c).
\end{equation}

\begin{remark}
A Noether symmetry is a symmetry of the forced Lagrangian system if and only if $\liedv{X^c}\alpha_L=0$. \label{Noether_L_symmetries}
\end{remark}

We have just discussed infinitesimal symmetries on $Q$, the so-called point-like symmetries \cite{MdL_DMdD_94}. We shall now cover symmetries which are not necessarily point-like, that is, vector fields on $TQ$.

\begin{definition}
A \emph{dynamical symmetry} of $\xi_{L,\beta}$ is a vector field $\tilde{X}$ on $TQ$ such that $[\tilde{X},\xi_{L,\beta}]=0$. A \emph{Cartan symmetry} is a vector field $\tilde{X}$ on $TQ$ such that $\tilde{X}(E_L)+\beta(\tilde{X})=0$
and $\liedv{\tilde{X}}\alpha_L=\dd f$.
\end{definition}
\begin{remark}
Let $X$ be a vector field on $Q$. Then
\begin{enumerate}[label=\roman*)]
 \item $X$ is a Lie symmetry if and only if $X^c$ is a dynamical symmetry.
 \item  $X$ is a Noether symmetry if and only if $X^c$ is a Cartan symmetry.
 \end{enumerate} 
\end{remark}

\begin{proposition}\label{prop_dynamical_sym}
If $\tilde X$ is a vector field on $TQ$ such that  
\begin{equation}
  \dd (\liedv{\tilde{X}}\alpha_L)=0,
\end{equation}
then $\tilde{X}$ is a dynamical symmetry if and only if
\begin{equation}
  \dd (\tilde{X} (E_L))=-\liedv{\tilde{X}}\beta.
\end{equation}
\end{proposition}

\begin{proposition}
A Cartan symmetry is a dynamical symmetry if and only if
\begin{equation}
  \contr{\tilde{X}}\dd \beta=0.
\end{equation}
\end{proposition}

\begin{proposition}
Let $\tilde X$ be a vector field on $TQ$ such that  
\begin{equation}
  \liedv{\tilde{X}}\alpha_L=\dd f.
\end{equation}
Then $\tilde{X}$ is a Cartan symmetry if and only if $f-(S\tilde{X})L$ is a constant of the motion.
\end{proposition}

The proofs are completely analogous to those for point-like symmetries. Notice that Theorem \ref{Noether_th} cannot be generalised for symmetries on $TQ$, since $[\xi_{L,\beta},\tilde{X}]$ is not a vertical vector field for a general $\tilde{X}$ on $TQ$.

\section{Symmetries and constants of the motion in the Hamiltonian description} \label{section_conserved_Hamiltonian}

Let $\alpha$ and $\hat{X}$ be a 1-form and a vector field on $T^*Q$, respectively. We say that $\alpha$ is a \emph{first integral} of $\hat{X}$ if $\alpha(\hat{X})=0$. Similarly, a function $F$ on $T^*Q$ is called a \emph{first integral} of $\hat{X}$ if $\dd F(\hat{X})=\hat{X}(F)=0$.

Let $(H,\gamma)$ be a Hamiltonian system with Hamiltonian function $H$ and external force $\gamma$. Let $X_{H,\gamma}$ be the corresponding Hamiltonian vector field. A first integral of $X_{H,\gamma}$ is called a \emph{constant of the motion} or a \emph{conserved quantity}.


 Let $F$ and $G$ be two functions on $T^*Q$. Let $X_F$ and $X_G$ be their corresponding Hamiltonian vector fields, namely $\contr{X_F}\omega_Q=\dd F$ and $\contr{X_G}\omega_Q=\dd G$.
The \emph{Poisson bracket} of $F$ and $G$ is given by
\begin{equation}
  \left\{F,G  \right\}=  \omega_Q(X_F,X_{G}).
\end{equation}
Let $\alpha$ and $\beta$ be 1-forms on $T^*Q$, with $X_\alpha$ and $X_\beta$ their corresponding Hamiltonian vector fields. Then their Poisson bracket is defined as
\begin{equation}
  \left\{\alpha,\beta  \right\}=-\contr{[X_\alpha,X_\beta]}\omega_Q.
\end{equation}
Clearly,
\begin{equation}
\begin{aligned}
  X_{H,\gamma}(F)&=\contr{X_{H,\gamma}}\dd F= \contr{X_{H,\gamma}}( \contr{X_F}\omega_Q)=-\contr{X_F}(\contr{X_{H,\gamma}}\omega_Q)\\
  &=-\omega_Q(X_H,X_F)-\gamma(X_F)= \left\{F,H  \right\}-\gamma(X_F),
\end{aligned}
\end{equation}
and hence $F$ is a constant of the motion if and only if
\begin{equation}
  \left\{F,H  \right\}=\gamma(X_F).
\end{equation}

\begin{proposition} 
If $\hat{X}$ is a vector field on $T^*Q$ such that $\liedv{\hat{X}}\alpha_Q$ is closed, then $\hat{X}$ commutes with $X_{H,\gamma}$ if and only if
\begin{equation}
  \dd (\hat{X} (H))=-\liedv{\hat{X}}\gamma.
\end{equation}
\end{proposition}

\begin{proposition}
Let $\hat{X}$ be a vector field on $T^*Q$ such that  
\begin{equation}
  \liedv{\hat{X}}\alpha_Q=\dd f.
\end{equation}
Then $\hat{X}(H)+\gamma(\hat{X})=0$ if and only if $f-\alpha_Q(\hat{X})$ is a constant of the motion. Additionally, $\hat{X}$ commutes with $X_{H,\gamma}$ if and only if 
\begin{equation}
  \contr{\hat{X}}\dd \gamma=0.
\end{equation}
\end{proposition}
Now suppose that $(L,\beta)$ is a Lagrangian system such that $H\circ \Leg = E_L$ and $\Leg^* \gamma=\beta$. Let $\tilde{X}$ be a vector field on $TQ$ and $\hat{X}$ the $\Leg$-related vector field on $T^*Q$. Then:
\begin{enumerate}[label=\roman*)]
\item $\hat{X}$ commutes with $X_{H,\gamma}$ if and only if $\tilde{X}$ is a dynamical symmetry of $(L,\beta)$.
\item $\liedv{\hat{X}}\alpha_Q=\dd f$ if and only if $\liedv{\tilde{X}}\alpha_L=\dd g$, where $g=f\circ \Leg$.
\item Suppose that $\liedv{\hat{X}}\alpha_Q=\dd f$. Then the following assertions are equivalent.
\begin{enumerate}[label=\alph*)]
\item $\hat{X}(H)+\gamma(\hat{X})=0$.
\item $f-\alpha_Q(\hat{X})$ is a conserved quantity.
\item $\tilde{X}(E_L)+\beta(\tilde{X})=0$.
\item $f\circ \Leg-\alpha_L(\tilde{X})$ is a conserved quantity.
\end{enumerate}

\end{enumerate}

\section{Rayleigh dissipation}\label{section_Rayleigh}
\subsection{Rayleigh dissipation function and 1-form}
Rayleigh \cite{rayleigh_1871} considers the hypothesis that there is a non-conservative force linear
on the velocities. This external force can be described as a semibasic 1-form on $TQ$ as follows:
\begin{equation}
R=R_{ij}(q)\dot q^i \dd q^j, \label{Rayleigh_function}
\end{equation} 
where $R_{ij}$ is symmetric. Of course, $R$ can be described as a bilinear form on $TQ$:
\begin{equation}
\begin{aligned}
R:\ &TQ\times TQ \to \RR,\\
& R(q^i,\dot q^i_1,\dot q^j_2)=R_{ij} \dot q^i_1\dot q^j_2, \label{Rayleigh_bilinear_form}
\end{aligned}
\end{equation}
or, in other words, a symmetric $(0,2)$-tensor $R$ on $Q$. Since $R$ is a $(0,2)$-tensor on $Q$, it defines a linear mapping 
\begin{equation}
\tilde{R}: TQ\to T^*Q
\end{equation} 
by
\begin{equation}
\tilde{R}(v_q)=\contr{v_q}R,
\end{equation}
that is,
\begin{equation}
\tilde{R}(v_q)(w_q)=R(v_q,w_q).
\end{equation}
Therefore
\begin{equation}
\tilde{R}(q^i,\dot q^i)=(q^i,R_{ij}(q)\dot q^i \dot q^j),
\end{equation}
so $\tilde{R}$ defines a semibasic 1-form $\bar{R}$ on $TQ$ given by
\begin{equation}
\bar{R}=R_{ij}(q)\dot q^i \dd q^j. \label{Rayleigh_1_form}
\end{equation}
In the literature \cite{gantmakher_70,goldstein_87} the Rayleigh dissipation function is defined as
\begin{equation}
  \mathcal{R}(q,\dot{q})=\frac{1}{2}R_{ij}(q)\dot{q}^{i} \dot{q}^{j}, \label{Rayleigh_dissipation_function}
\end{equation}
so we can write
\begin{equation}
  \bar{R}= \frac{\partial \mathcal{R}} {\partial \dot{q}^{i}}\dd q^{i}= S^{*}(\dd \mathcal{R}).
\end{equation}
Notice that external forces of the form $S^*(\dd \mathcal{F})$, for some function $\mathcal{F}$, are quite more general than the form \eqref{Rayleigh_function} originally proposed by Rayleigh \cite{rayleigh_1871}. That is, we do not need to require $\mathcal{R}$ to be of the form \eqref{Rayleigh_dissipation_function}. In fact, more general dissipation functions are studied in Refs.~\cite{lurie_02,minguzzi_15}. This function $\mathcal{R}$ can be physically interpreted as a potential that depends on the velocities from which the external force is derived.

\begin{remark}
    $\mathcal{R}$ and $\tilde{\mathcal{R}}=\mathcal{R}+f$ define the same 1-form $\bar{R}$, for $f:Q\to \RR$ arbitrary.
\end{remark}


We shall now consider a Lagrangian system with hyperregular Lagrangian function $L$, and which is subject to an external force linear on the velocities. Suppose that the force can be described through the Rayleigh dissipation function $\mathcal{R}$. Then the equations of motion of the system are the integral curves of the vector field $\xi_{L,\bar{R}}$, given by
\begin{equation}
\contr{\xi_{L,{\bar{R}}}}\omega_L=\dd E_L+\bar{R}.
\end{equation} 
Let $\xi_{\bar{R}}$ be the vector field given by
\begin{equation}
\contr{\xi_{{\bar{R}}}}\omega_L={\bar{R}},
\end{equation}
then
\begin{equation}
\xi_{L,{\bar{R}}}=\xi_L+\xi_{\bar{R}},
\end{equation}
where $\xi_L$ is the vector field given by Eq.~\eqref{Lagrangian_VF}.
We have
\begin{equation}
\xi_{\bar{R}}=-R_{ik} \dot{q}^k W^{ij}\parder{}{\dot{q}^j},
\end{equation}
where $W^{ij}$ is the inverse of the Hessian matrix of the Lagrangian \eqref{Hessian_Lagrangian}. Then $\xi_{L,\bar{R}}$ is a SODE in the sense of Eq.~\eqref{SODE}, and the equations of motion of the system are
\begin{equation}
\der{}{t}\left(\parder{L}{\dot{q}^i}\right)-\parder{L}{q^i}=-R_{ij}(q)\dot{q}^j=-\frac{\partial \mathcal{R}} {\partial \dot{q}^i}.
\end{equation} 
It is easy to see that
\begin{equation}
  \xi_{L,\beta}(E_L)+\Delta (\mathcal{R})=0,
\end{equation}
where $\Delta$ is the Liouville vector field. In particular, if $\mathcal{R}$ is of the form \eqref{Rayleigh_dissipation_function}, then $\Delta (\mathcal{R})=2\mathcal{R}$.

We can also consider the Hamiltonian formalism for the Rayleigh dissipation. Indeed, since we have assumed $L$ to be hyperregular, we can always define the external force $\hat{R}$ on $T^*Q$ by
\begin{equation}
\bar{R}=\Leg^*\hat{R},
\end{equation}
and consider the Hamiltonian function given by Eq.~\eqref{Hamiltonian_Lagrangian}. Locally $\hat{R}$ can be written as
\begin{equation}
\hat{R}={R}_{ij}(q)p_i \dd q^j
\end{equation}
Then the equations of motion of the system are the integral curves of the vector field $X_{H,\hat{R}}$, given by
\begin{equation}
\contr{X_{H,\hat{R}}}\omega_Q=\dd H+\hat{R}.
\end{equation}
If $Z_{\hat{R}}$ is the vector field defined by
\begin{equation}
\contr{Z_{\hat{R}}}\omega_Q=\hat{R},
\end{equation}
then we have
\begin{equation}
X_{H,\hat{R}}=X_H+Z_{\hat{R}},
\end{equation}
where $X_H$ is the Hamiltonian vector field given by Eq.~\eqref{Hamiltonian_VF}. In canonical coordinates, 
$X_H$ and $\hat{R}$ are given by Eqs.~\eqref{Hamiltonian_VF_local} and \eqref{Rayleigh_1_form}, respectively, and we have
\begin{gather}
Z_{\hat{R}}=-{R}_{ij}(q)p_i\parder{}{p_j},\\
X_{H,\hat{R}}=\parder{H}{p_i}\parder{}{q^i}-\left(\parder{H}{q^i}+{R}_{ij}(q) p_j\right)\parder{}{p_i},
\end{gather}
where we have made use of the fact that $R_{ij}$ is symmetric. Thus, the equations of motion are
\begin{equation}
\begin{aligned}
&\der{q^i}{t}=\parder{H}{p_i},\\
&\der{p_i}{t}=-\left(\parder{H}{q^i}+R_{ij}(q) p_j\right).
\end{aligned}
\end{equation}
As we have shown in Section \ref{section_morphisms}, given the semibasic 1-form $\bar{R}$, we can define the following morphism of fibred bundles:

\begin{center}
\begin{tikzcd}
D_{\bar{R}}: TQ \arrow[rr] \arrow[rd, "\tau_q"'] &   & T^*Q \arrow[ld, "\pi_Q"] \\
                                                          & Q &                         
\end{tikzcd},
\end{center}
\begin{equation}
    \left\langle D_{\bar{R}}(v_q),w_q\right\rangle={\bar{R}}(v_q)(u_{w_q}),
\end{equation}
for every $v_q,w_q\in T_qQ,\ u_{w_q}\in T_{w_q}(TQ)$, with $T\tau_Q(u_{w_q})=w_q$. In local coordinates, 
we have
\begin{equation}
    D_{\bar{R}}(q^i,\dot q^i)=\left(q^i,R_{ij}(q)\dot{q}^j\right)
\end{equation}
\subsection{Constants of the motion for Rayleigh dissipation}
We shall now consider the case in which the external force is derived from a dissipation function $\mathcal{R}$ (not necessarily quadratic in the velocities).

\begin{lemma}\label{lemma_rayleigh_potential} Consider a semibasic 1-form $\bar{R}$ on $TQ$ given by 
\begin{equation}
  \bar{R} = S^*(\dd \mathcal{R})
\end{equation}
for some function $\mathcal{R}:TQ\to \RR$, where $S^*$ is the adjoint of the vertical endomorphism. Then, for each vector field $X$ on $Q$
\begin{flalign}
  &\bar{R}(X^{c})=X^{v}(\mathcal{R}),\\
  &\liedv{X^c}  \bar R = S^* (\dd (X^c (\mathcal R))). \label{liedv_Rayleigh}
\end{flalign}
Similarly, for each vector field $\tilde{X}$ on $TQ$,
\begin{equation}
  \bar{R}(\tilde{X}) = (S\tilde{X}) (\mathcal{R}).
\end{equation}
\end{lemma}
\begin{proof}
  Indeed, 
  \begin{equation}
  \bar{R}(\tilde{X})
  =\contr{\tilde{X}} \bar{R}
  =\contr{\tilde{X}}(S^{*}\dd \mathcal{R})
  =\contr{S \tilde{X}}\dd \mathcal{R}
  =(S\tilde{X}) (\mathcal{R}).
\end{equation}
In particular, $SX^c=X^v$. Eq.~\eqref{liedv_Rayleigh} can be shown by direct computation in bundle coordinates.
\end{proof}

\begin{proposition}
  Let $X$ be a vector field on $Q$. Then $X^c(L)=X^{v}(\mathcal{R})$ if and only if $X^{v}(L)$ is a constant of the motion.
\end{proposition}

\begin{example}[Fluid resistance]
Consider a body of mass $m$ moving through a fluid that fully encloses it. For the sake of simplicity, suppose that the motion takes place along one dimension. 
Then the drag force \cite{batchelor_00,falkovich_11} is given by
\begin{equation}
  \bar{R}
  =\frac{1}{2} \rho~CA \dot{q}^2 \dd q,
\end{equation}
where $C$ is a dimensionless constant depending on the body shape, $\rho$ is the mass density of the fluid and $A$ is the area of the projection of the object on a plane perpendicular to the direction of motion. For the sake of simplicity, suppose that the density is uniform, and then $k=CA\rho/2$ is constant.
The dissipation function is thus
\begin{equation}
  \mathcal{R}=\frac{k}{3}\dot{q}^3.
\end{equation}
If the body is not subject to forces besides the drag, its Lagrangian is $L=m\dot{q}^2/2$. Consider the vector field $X=e^{k/m q}\partial/\partial q$. We can verify that $X^c(L)=X^v(\mathcal{R})$, so $X^v(L)=me^{k/mq}\dot{q}$ is a constant of the motion. In particular, when $k\to 0$ we recover the conservation of momentum.
\end{example}

\begin{proposition}
If $\liedv{X^c}\alpha_L$ is closed, then $X$ is a Lie symmetry of $(L,\bar{R})$ if and only if 
\begin{equation}
  \dd(X^c(E_L))=-S^* (\dd (X^c \mathcal R)).
\end{equation}
\end{proposition}
\begin{proposition}
If $\liedv{X^c}{\alpha_L}=\dd f$ for some function $f:TQ\to \RR$, then the following statements are equivalent:
\begin{enumerate}[label=\roman*)]
\item $X$ is a Noether symmetry.
\item $X^c(E_L)+X^v(\mathcal{R})=0$.
\item $f-X^v(L)$ is a constant of the motion.
\end{enumerate}
Moreover, a Noether symmetry is a Lie symmetry if and only if $\contr{X^c}\dd \bar{R}=0$.
\end{proposition}

 Let $\tilde{X}$ be a vector field on $TQ$.
If $\liedv{\tilde{X}}\alpha_L$ is closed, then $\tilde{X}$ is a dynamical symmetry if and only if
\begin{equation}
  \dd(\tilde{X}(E_L)+(S\tilde{X})(\mathcal{R}))=-\contr{\tilde{X}} \dd \bar{R}.
\end{equation}
\begin{proposition}
 If $\liedv{\tilde{X}}{\alpha_L}=\dd f$, then the following statements are equivalent:
 \begin{enumerate}[label=\roman*)]
 \item  $\tilde{X}$ is a Cartan symmetry.
 \item $\tilde{X}(E_L)+(S\tilde{X})(\mathcal{R})=0.$
 \item $f-(S\tilde{X})(L)$ is a conserved quantity.
 \end{enumerate}
\end{proposition}

We shall now cover some examples proposed in Ref.~\cite{minguzzi_15} and obtain their constants of motion.

\begin{example}[A rotating disk]
Let us consider a disk of mass $m$ and radius $r$ placed on a horizontal surface. Let $\varphi$ be the angle of rotation of the disk with respect to a reference axis. The Lagrangian of the disk is $L=T=mr^{2}\dot{\varphi}^{2}/4$ and its Rayleigh dissipation function is $\mathcal{R}=\mu mgr\dot\varphi/2$. The Poincaré-Cartan 1-form is $\alpha_L=mr^2 \dot \varphi/2\  \dd \varphi$. The external force is $\bar{R}=\mu mgr/2\ \dd \varphi$.

 Consider the vector field $\tilde{X}=r\dot{\varphi}\partial/\partial \varphi+\mu g \partial/\partial \dot{\varphi}$. Clearly, $\tilde{X}(E_L)=\tilde{X}(L)=(S\tilde{X})(\mathcal{R})$. We have that
 \begin{equation}
   \liedv{\tilde{X}}\alpha_L=
    \frac{\mu mgr^2}{2} \dd \varphi+
   \frac{mr^3}{2}\dot \varphi \dd \dot \varphi=\dd f,
 \end{equation}
 where
 \begin{equation}
   f=\frac{\mu mgr^2}{2}\varphi +\frac{mr^3}{4}\dot \varphi^2
 \end{equation}
 modulo a constant, and $(S\tilde{X})(L)=mr^{3} \dot{\varphi}^{2}/2$,
  so
  \begin{equation}
    f-(S\tilde X) (L)=\frac{\mu mgr^2}{2}\varphi -\frac{mr^3}{4}\dot \varphi^2
  \end{equation}
  is a constant of the motion. Since $\bar{R}$ is closed, $\contr{\tilde{X}}\dd \bar{R}=0$ is trivially satisfied, so $\tilde{X}$ is a dynamical symmetry as well as a Cartan symmetry. 

  However, since $\bar{R}$ is closed, it is not strictly an external force. In fact, the Lagrangian
  \begin{equation}
    \tilde{L}=L+\frac{\mu mgr}{2} \varphi
  \end{equation}
  leads to the same equations of motion as $L$ with the external force $\bar{R}$.
\end{example}

\begin{example}[The rotating stone polisher]
Consider a system formed by two concentric rings of the same mass $m$ and radius $r$, which are placed over a rough surface, and rotate in opposite directions. Let $(x,y)$ be the position of the centre and $\theta$ the orientation of the machine. Let $\omega$ be the angular velocity of the rings. The Rayleigh dissipation function is given by
\begin{equation}
  \mathcal{R}=2\mu mg r\omega +\frac{\mu mg}{2r\omega}(\dot{x}^{2}+\dot{y}^{2}), 
\end{equation}
and the Lagrangian is $L=T=m(\dot{x}^{2}+\dot{y}^{2} +r^{2}\dot{\theta}^{2}+r^{2}\omega^{2})$. The Poincaré-Cartan 1-form is $\alpha_L=2m(\dot x\dd  x+ \dot y \dd  y +r^2 \dot \theta \dd  \theta)$
, and the external force is
\begin{equation}
  \bar{R}=\frac{\mu mg}{r\omega}(\dot x\dd x + \dot y \dd y).
\end{equation}
 Let $\tilde{X}^{(1)}=2r\omega \partial/\partial x +\mu g \partial/ \partial \dot{x}$ and $\tilde{X}^{(2)}=2r\omega \partial/\partial y +\mu g  \partial/ \partial \dot{y}$. We can check that $\tilde{X}^{(i)}(E_L)=\tilde{X}^{(i)}(L)=(S\tilde{X}^{(i)})(\mathcal{R})$ (for $i=1,2$). We have that $\liedv{\tilde{X}^{(i)}}\alpha_L=\dd f_i$ for $f_1=2\mu mgx$ and $f_2=2\mu mgy$, along with $(S\tilde{X}^{(1)})(L)=4mr \dot x$ and $(S\tilde{X}^{(2)})(L)=4mr \dot y$, so $2mr\omega \dot{x}-\mu mgx$ and $2mr\omega \dot{y}-\mu mgy$ are constants of the motion.
\end{example}

\section{Momentum map and reduction}\label{section_reduction}

It is well-known that if a $d$-dimensional symmetry group is acting over a physical system, then the number of independent degrees of freedom is reduced by $d$. In other words, $Q$ is reduced by $d$ dimensions, so $TQ$ and $T^*Q$ are reduced by $2d$ dimensions. Therefore $2d$ variables can be eliminated from the equations of motion. This fact can be exploited in a systematic way by means of the procedure known as \emph{reduction}, which is due to Marsden and Weinstein \cite{marsden_weinstein_74,abraham_marsden_08}.

Let $G$ be a Lie group acting on $Q$ and consider the lifted action to $TQ$ using tangent prolongation, that is, if $\Phi_g:Q\to Q$ is the diffeomorphism given by $\Phi_g(q)=gq$ for each $g\in G$ and $q\in Q$, then the lifted action is defined by
\begin{equation}
  T\Phi_g: TQ\to TQ.
\end{equation}
In what follows, we shall assume every group action considered to be free and proper.
Let $\mathfrak{g}$ the the Lie algebra of $G$ and $\mathfrak{g}^*$ its dual. Let $L:TQ\to \RR$ be a Lagrangian function subjected to an external force $\beta$.
 Suppose that the $G$-action leaves $L$ invariant, and hence $\alpha_L$ and $\omega_L$ are invariant. Then the natural momentum map \cite{abraham_marsden_08},
\begin{equation}
\begin{aligned}
  J:TQ&\to \mathfrak{g}^*,\\
  J(v_q)(\xi)&=\alpha_L(v_q) \left(\xi_Q^c(v_q)  \right),
\end{aligned}
\end{equation}
is equivariant and Hamiltonian. For each $\xi\in \mathfrak{g}$ and $v_\in TQ$, $J\xi:TQ\to \RR$ is the function given by
\begin{equation}
  J\xi (v_q) = \left\langle J(v_q),\xi  \right\rangle   \label{J_xi_function}.
\end{equation} 
\begin{lemma} \label{lemma_vector_space}

Let $\xi \in \mathfrak{g}$.
Then 
\begin{enumerate}[label=\roman*)]
\item 
 $J\xi$ is a conserved quantity for $\xi_{L,\beta}$ if and only if
 \begin{equation}
   \beta(\xi_Q^c)=0 \label{lemma_J_conserved}
 \end{equation}
 \item If the previous equation holds, then $\xi$ leaves $\beta$ invariant if and only if
 \begin{equation}
   \contr{\xi_Q^c}\dd \beta = 0. \label{lemma_beta_invariant}
 \end{equation}
\end{enumerate}
In addition, the vector subspace of $\mathfrak{g}$ given by
\begin{equation}
  \mathfrak{g}_\beta = \left\{\xi \in \mathfrak{g} \mid
    \beta(\xi_Q^c)=0,\ \contr{\xi_Q^c}\dd \beta = 0
    \right\}
\end{equation}
is a Lie subalgebra of $\mathfrak{g}$.
\end{lemma}
\begin{proof}
\begin{enumerate}[label=\roman*)]
\item 
We have that
\begin{equation}
  J\xi= \alpha_L (\xi_Q^c) = \contr{\xi_Q^c}\alpha_L,
\end{equation}
so
\begin{equation}
 \dd(J\xi) = \dd (\contr{\xi_Q^c}\alpha_L) 
 = \liedv{\xi_Q^c}\alpha_L - \contr{\xi_Q^c} \dd \alpha_L
 = \contr{\xi_Q^c}\omega_L 
 \label{Hamiltonian_vector_momentum_map}.
\end{equation}

Contracting this equation with $\xi_{L,\beta}$, one gets
\begin{equation}
  \contr{\xi_{L,\beta}}(\dd(J\xi))=\xi_{L,\beta}(J\xi),
\end{equation}
on the left-hand side, and
\begin{equation}
  \contr{\xi_{L,\beta}}\contr{\xi_Q^c}\omega_L
  =-\contr{\xi_Q^c}\contr{\xi_{L,\beta}}\omega_L
  =-\contr{\xi_Q^c}(\dd E_L+\beta)=-\xi_Q^c(E_L)-\beta(\xi_Q^c),
\end{equation}
on the right-hand side.

Thus $J\xi$ is a conserved quantity for $\xi_{L,\beta}$ if and only if
\begin{equation}
  \xi_Q^c(E_L)+\beta(\xi_Q^c)=0. \label{conservation_dynamics_algebra}
\end{equation}
Now observe that
\begin{equation}
\begin{aligned}
  \xi_Q^c(E_L)&=\xi_Q^c(\Delta(L))-\xi_Q^c(L)=\xi_Q^c(\Delta(L))
  =\liedv{\xi_Q^c}(\Delta(L))=\liedv{\xi_Q^c}(\contr{\Delta}\dd L)\\
  &=\contr{[\xi_Q^c,\Delta]}\dd L+\contr{\Delta}(\liedv{\xi_Q^c} \dd L)
  =\contr{[\xi_Q^c,\Delta]}\dd L
  =[\xi_Q^c,\Delta](L),
\end{aligned}
\end{equation}
since $\xi_Q^c(L)=0$ by the $G$-invariance of $L$, but $[\xi_Q^c,\Delta]=0$, and thus
\begin{equation}
  \xi_Q^c(E_L)=0 
\end{equation}
for each $\xi\in \mathfrak{g}$, that is, $E_L$ is $G$-invariant. By Eq.~\eqref{conservation_dynamics_algebra}, $J\xi$ is a conserved quantity for $\xi_{L,\beta}$ if and only if
\begin{equation}
  \beta(\xi_Q^c) = 0. \label{conservation_dynamics_algebra_simplified}
\end{equation}
\item For each $\xi\in \mathfrak{g}_\beta$, we have that
\begin{equation}
  \liedv{\xi_Q^c} \beta = \dd (\contr{\xi_Q^c} \beta) + \contr{\xi_Q^c} \dd \beta = \dd (\beta(\xi_Q^c)) + \contr{\xi_Q^c} \dd \beta.
\end{equation}
If Eq.~\eqref{conservation_dynamics_algebra_simplified} holds, then $\beta$ is $\mathfrak{g}_\beta$-invariant (i.e., $\liedv{\xi_Q^c} \beta = 0$) if and only if
\begin{equation}
  \contr{\xi_Q^c} \dd \beta = 0
\end{equation}
\end{enumerate}

For $\mathfrak{g}_\beta$ being a Lie subalgebra it is necessary and sufficient that $[\xi,\eta]\in \mathfrak{g}_\beta$ for each $\xi,\eta\in \mathfrak{g}_\beta$. Since $\xi\in \mathfrak{g}\mapsto \xi_Q\in \mathfrak{X}(Q)$ is a Lie algebra antihomomorphism \cite{ortega_ratiu_04}, this is equivalent to
\begin{flalign}
  &\beta \left([\xi_Q,\eta_Q]^c  \right)=0, \label{subalgebra_condition_1}\\
  &\contr{[\xi_Q,\eta_Q]^c} \dd\beta=0, \label{subalgebra_condition_2} 
\end{flalign}
but $[\xi_Q,\eta_Q]^c=[\xi_Q^c,\eta_Q^c]$, since the complete lift is a morphism between Lie algebras. Then
\begin{equation}
\begin{aligned}
  \beta \left([\xi_Q,\eta_Q]^c  \right) &=  \beta \left([\xi_Q^c,\eta_Q^c]  \right) 
  = \contr{[\xi_Q^c,\eta_Q^c]}\beta
  = \liedv{\xi_Q^c} \contr{\eta_Q^c} \beta 
    - \contr{\eta_Q^c} \liedv{\xi_Q^c}  \beta\\
  &= \xi_Q^c (\beta(\eta_Q^c)) - \eta_Q^c (\beta(\xi_Q^c)) 
  - \contr{\eta_Q^c} (\contr{\xi_Q^c} \dd \beta)
  = 0 
    , 
\end{aligned}
\end{equation}
by Eqs.~\eqref{lemma_beta_invariant} and \eqref{lemma_J_conserved}.
Similarly,
\begin{equation}
\begin{aligned}
   \contr{[\xi_Q,\eta_Q]^c} \dd\beta
   &= \contr{[\xi_Q^c,\eta_Q^c]} \dd\beta
   = \liedv{\xi_Q^c} \contr{\eta_Q^c} \dd \beta 
    - \contr{\eta_Q^c}\liedv{\xi_Q^c}\dd \beta\\
    &= \liedv{\xi_Q^c} \contr{\eta_Q^c} \dd \beta 
    - \contr{\eta_Q^c}\dd \liedv{\xi_Q^c} \beta=0.
\end{aligned}
\end{equation}
\end{proof}
It is worth mentioning that our Lemma \ref{lemma_vector_space} i) was previously obtained by Marsden and West \cite[Theorem 3.1.1]{marsden_west_01}, albeit from a variational approach.

\begin{corollary}
For each $\xi\in \mathfrak{g}_\beta$, $\xi_Q^c$ is a Noether symmetry and it is a symmetry of the forced Lagrangian system.
\end{corollary}

\begin{proof}
Since $\alpha_L$ is invariant, 
\begin{equation}
  \liedv{\xi_Q^c}\alpha_L=0
\end{equation}
for each $\xi \in \mathfrak{g}$. In combination with Eq.~\eqref{conservation_dynamics_algebra}, this implies that $\xi_Q^c$ is a Noether symmetry. By Remark~\ref{Noether_L_symmetries}, it is also a symmetry of the forced Lagrangian system.
\end{proof}

\begin{theorem}\label{theorem_reduction}
 Let  $G_\beta\subset G$ be the Lie subgroup generated by $\mathfrak{g}_\beta$ and $J_\beta:TQ\to \mathfrak{g}_\beta^*$ the reduced momentum map. Let $\mu\in \mathfrak{g}_\beta^*$ be a regular value of $J_{\beta}$ and $(G_\beta)_\mu$ the isotropy group in $\mu$.
 Then 
\begin{enumerate}[label=\roman*)]
\item $J_\beta^{-1}(\mu)$ is a submanifold of $TQ$ and $\xi_{L,\beta}$ is tangent to it. 
\item The quotient space $M_\mu\coloneqq J_\beta^{-1}(\mu)/(G_\beta)_\mu$ is endowed with an induced symplectic structure $\omega_\mu$, namely 
\begin{equation}
  \pi_\mu^*\omega_\mu = \incl_\mu^*\omega_L,
\end{equation}
where $\pi_\mu:J_\beta^{-1}(\mu)\to M_\mu$ and $\incl_\mu:J_\beta^{-1}(\mu)\hookrightarrow TQ$ denote the projection and the inclusion, respectively. 
\item $L$ induces a function $L_\mu:M_\mu\to \RR$ defined by
\begin{equation}
  L_\mu \circ \pi_\mu = L\circ \incl_\mu.
\end{equation}
Moreover, we can introduce a function $E_{L_\mu}:M_\mu\to \RR$, given by $E_{L_\mu}=\Delta_\mu(L_\mu)-L_\mu$, which satisfies
\begin{equation}
  E_{L_\mu} \circ \pi_\mu = E_L\circ \incl_\mu. \label{reduced_energy}
\end{equation}
\item $\beta$ induces a reduced semibasic 1-form $\beta_\mu$ on $M_\mu$, given by
\begin{equation}
  \pi_\mu^* \beta_\mu=\incl_\mu^*\beta.
\end{equation}
\end{enumerate}
\end{theorem}

\begin{proof}
For a proof of the first three assertions see Refs.~\cite{marsden_montgomery_ratiu_90,abraham_marsden_08,ortega_ratiu_04}.
Observe that
\begin{equation}
\begin{aligned}
\Delta(L)\circ \incl_\mu
&= \incl_\mu^*\Delta(L)
=\incl_\mu^* (\contr{\Delta}\dd L)
=\contr{\incl_{\mu*}\Delta} (\incl_\mu^*\dd L)
=(\incl_{\mu*}\Delta)(\incl_{\mu}^* L)\\
&=(\incl_{\mu*}\Delta)(L\circ \incl_{\mu})
=(\incl_{\mu*}\Delta)(L_\mu\circ \pi_{\mu})
= (\pi_{\mu*}\Delta_\mu) (L_\mu\circ \pi_{\mu})\\
&= \pi_\mu^* (\Delta_\mu(L_\mu)) 
= \Delta_\mu(L_\mu) \circ \pi_\mu
,
\end{aligned}
\end{equation}
where $\Delta$ is the Liouville vector field on $TQ$ and $\Delta_\mu$
is a $\pi_\mu$-related vector field on $M_\mu$, namely
\begin{equation}
 \pi_{\mu*}\Delta_\mu = \incl_{\mu*}\Delta.
\end{equation}
Then we can introduce a function $E_{L_\mu}:M_\mu\to \RR$, given by $E_{L_\mu}=\Delta_\mu(L_\mu)-L_\mu$, which satisfies Eq.~\eqref{reduced_energy}. Since $\beta$ is $(G_\beta)_\mu$-invariant, it induces a reduced semibasic 1-form $\beta_\mu$ on $M_\mu$.
\end{proof}

\begin{corollary}
The vector field $\xi_{L_\mu,\beta_\mu}$, defined by
 \begin{equation}
   \contr{\xi_{L_\mu,\beta_\mu}}\omega_\mu =\dd E_{L_\mu}+\beta_\mu,
 \end{equation}
 determines the dynamics on $M_\mu$. It is $\pi_\mu$-related to $\xi_{L,\beta}$.
\end{corollary}

\begin{remark}\label{rayleight_reduction}
In the Rayleigh dissipation case $\beta = S^*(\dd \mathcal{R})$, according to Lemma~\ref{lemma_rayleigh_potential}, one can equivalently define $\mathfrak{g}_\beta\equiv \mathfrak{g}_\mathcal{R} $ as
\begin{equation}
     \mathfrak{g}_\mathcal{R} = \left\{\xi \in \mathfrak{g} \mid
    \xi_Q^v(\mathcal{R}) = 0, S^* \dd \xi_Q^c(\mathcal{R}) = 0
    \right\},
\end{equation}
where the condition $S^* \dd \xi_Q^c(\mathcal{R}) = 0$ means that $\xi_Q^c(\mathcal{R})$ is basic: it does not depend on $\dot{q}^i$. If additionally $\mathcal{R}$ is $\mathfrak{g}_\mathcal{R}$-invariant, i.e., $\xi_Q^c(\mathcal{R})=0$ for each $\xi\in \mathfrak{g}_\mathcal{R}$, then it induces a dissipation function $\mathcal{R}_\mu:M_\mu\to \RR$ given by
\begin{equation}
  \mathcal{R}_\mu \circ \pi_\mu = \mathcal{R} \circ \incl_\mu.
\end{equation}
\end{remark}

\begin{remark}[Reconstruction of the dynamics]
Knowing the integral curves on $M_\mu$, we want to obtain the integral curves on $J^{-1}(\mu)$.
Let $c(t)$ and $[c(t)]$ be the integral curves of $\xi_{L,\beta}$ and $\xi_{L_\mu,\beta_\mu}$, respectively, with $c(0)=p_0$. Let $d(t)\in J^{-1}(\mu)$ be a smooth curve such that $d(0)=p_0$ and $[d(t)]=[c(t)]$. We can write
\begin{equation}
  c(t)=\Phi_{g(t)}(d(t)),\label{integral_curve}
\end{equation}
for $g(t)\in (G_\beta)_\mu$. Then we have to find $g(t)$ in order to express $c(t)$ in terms of $[c(t)]$. Now \cite{abraham_marsden_08}
\begin{equation}
\begin{aligned}
  \xi_{L,\beta}(c(t))
  &=c'(t)\\
  &=T\Phi_{g(t)}(d(t))(d'(t)) + T\Phi_{g(t)}(d(t))(TL_{g(t)^{-1}}(g'(t)))_Q^c(d(t)),
\end{aligned}
\end{equation}
and using the $\Phi_g$-invariance of $\xi_{L,\beta}$ one gets
\begin{equation}
  \xi_{L,\beta}(d(t))
  =d'(t)+(TL_{g(t)^{-1}}(g'(t)))_Q^c(d(t)).
\end{equation}
In order to solve this equation, we first solve the algebraic problem
\begin{equation}
  \xi_Q^c(d(t))=\xi_{L,\beta}(d(t))-d'(t),
\end{equation}
for $\xi(t)\in\mathfrak{g}_\beta$, and then solve
\begin{equation}
  g'(t)=TL_{g(t)}\xi(t)
\end{equation}
for $g(t)$. The integral curve sought is given by Eq.~\eqref{integral_curve}.
\end{remark}

\begin{example}[Angular momentum]
    Consider $Q=\RR^n\setminus \{0\}$, a Lagrangian function $L$ on $T Q$ that is spherically symmetric, say $L(q,\dot{q})= L(\lVert q \rVert,\lVert{\dot{q}}\rVert)$. 
    Consider the Lie group $G=\mathrm{SO}(n) = \{O \in \RR^{n\times n} \mid {O}^t O = \Id, \det(O)=1 \}$ acting by rotations on $Q$. The action can be lifted to $TQ$ via the tangent lift. Explicitly, for $O\in \mathrm{SO}(3)$ we let
    \begin{equation}
        \begin{aligned}
            g_O :  Q  &\to  Q, \\
            (q) &\mapsto (O \cdot q)\\
            Tg_O : T Q  &\to T Q, \\
            (q,\dot{q}) &\mapsto (O \cdot q, O \cdot \dot{q}).
        \end{aligned}
    \end{equation}
The group $\mathrm{SO}(n)$ acts freely and properly. The Lie algebra of the group is given by $\mathfrak{g}=\mathfrak{so}(n) = \{o \in \RR^{n\times n} \mid {o}^t + o = 0\}$.
In the case $n=3$ this algebra $\mathfrak{g}$ will be identified with the algebra of three dimensional vectors with the cross product by taking
\begin{equation}
    \begin{pmatrix}
        0 & -\xi_3 & \xi_2\\
        \xi_3 & 0 & -\xi_1\\
        -\xi_2 & \xi_1 & 0\\
    \end{pmatrix}
    \mapsto 
    \begin{pmatrix}
        \xi_1 \\ \xi_2 \\ \xi_3
    \end{pmatrix}.
\end{equation}
The infinitesimal generator of $\xi\in \mathfrak{g}$ is given by 
\begin{equation}
\begin{split}
    \xi_Q (q) &= (\xi \times q),\\
    \xi_Q^c(q,\dot{q}) &= (\xi \times q,\xi \times \dot{q}),\\
    \xi_Q^v(q,\dot{q}) &= (0,\xi \times \dot{q}).
\end{split}
\end{equation}
One can identify $\mathfrak{g}$ with $\mathfrak{g}^*$ by using the inner product on $\RR^3$. The moment map is then given by \cite[Example~4.2.15]{abraham_marsden_08}
\begin{equation}
    J(q,\dot{q}) = q \times \dot{q}. 
\end{equation}
Identifying $\mathfrak{g}^*\simeq \RR^3$ one sees that the coadjoint actions of $G$ is the usual one (by rotations). Let $\mu \in \mathfrak{g}^*$, $\mu\neq 0$ then the isotropy group $G_\mu \simeq S^1$ of $\mu$ under the coadjoint action, which are the rotations around the axis $\mu$.

We look for Rayleigh potentials $\mathcal{R}$ such that $\mathfrak{g}_\mathcal{R} = \mathfrak{g}$. The condition that $\xi_Q^v (\mathcal{R})=0$ implies that $\mathcal{R}$ is spherically symmetric on the velocities. Then, the condition that $\xi_Q^c(q,\dot{q})$ is semi-basic means that the terms which are not spherically symmetric on the positions cannot involve the velocities. That is:
\begin{equation}
    \mathcal{R} = A(q) + B(\lVert{q}\rVert, \lVert \dot{q} \rVert).
\end{equation}
Without loss of generality, one can take $\mu=(0,0,\mu_0)$. Hence, if $(q,v) \in J^{-1}(\mu)$, both $q$ and $\dot{q}$ lie on the $xy$-plane. Moreover, they must satisfy the equation $\dot{q}^1 p_2-\dot{q}^1 q^2 = \mu_0$. We can finally apply Theorem~\ref{theorem_reduction}, to our system and find out that $(M_\mu, L_\mu)$, which is a Hamiltonian system over a $2$-dimensional manifold~\cite[Example~4.3.4]{abraham_marsden_08}.
\end{example}

\section*{Acknowledgements}
The authors acknowledge financial support from the Spanish Ministry of Science and Innovation (MICINN), under grants PID2019-106715GB-C21, ``Severo Ochoa Programme for Centres of Excellence in R\&D'' (CEX2019-000904-S) and from the Spanish National Research Council (CSIC), through the ``Ayuda extraordinaria a Centros de Excelencia Severo Ochoa'' (20205-CEX001). Manuel Laínz wishes to thank MICINN and the Institute of Mathematical Sciences (ICMAT) for the FPI-Severo
Ochoa predoctoral contract PRE2018-083203. 
Asier López-Gordón would like to thank CSIC for its financial support through the JAE Intro scholarship JAEINT\_20\_01494. He would also like to acknowledge a renewal of the scholarship by ICMAT. The constructive corrections of the referee are also appreciated.

\section*{Data Availability}
Data sharing is not applicable to this article as no new data were created or analyzed in this study.


\printbibliography
\end{document}